\documentclass{article}
\usepackage[utf8]{inputenc}
\usepackage{amsmath}
\usepackage{longtable, booktabs}
\usepackage[lined,boxed,commentsnumbered,english]{algorithm2e}
\SetKwProg{Fn}{Procedure}{}{}
\usepackage{graphicx}
\usepackage{cleveref}
\usepackage{xfrac}
\usepackage{adjustbox}
\usepackage{csquotes}
\usepackage{subcaption}
\usepackage{graphicx}
\usepackage{geometry}
\usepackage{multirow}
\usepackage{color}
\usepackage{siunitx}
\usepackage{tikz} 
\usepackage{pgfplots}
\pgfplotsset{compat=newest} 
\usepgfplotslibrary{units} 

\usepackage{authblk}

\usepackage{graphicx}
\usepackage{caption}

\usepackage{hyperref}
\hypersetup{
    colorlinks,
    citecolor=black,
    filecolor=black,
    linkcolor=black,
    urlcolor=blackhttps://www.overleaf.com/project/5b7cecf054fb8e520df8c434
}
 
 \usepackage{amsthm}
 
 \newtheorem{theorem}{Theorem}[section]
 \newtheorem{lemma}[theorem]{Lemma}
 \newtheorem{proposition}[theorem]{Proposition}
 
 \theoremstyle{definition}
 
 \usepackage{braket}

\providecommand{\keywords}[1]{\textbf{Keywords: } #1}
\SetKwFor{ForEach}{foreach}{do}{}
\SetKwRepeat{When}{when}{end}
\SetKwIF{If}{ElseIf}{Else}{if}{then}{else if}{else}{}
\SetKwInOut{Input}{Input}
\SetKwInOut{Output}{Output}
\SetKwRepeat{Event}{Event}{end}
\SetKwFor{Step}{Step}{}{}

\DeclareMathOperator*{\argmin}{argmin}

\newcommand{\init}{\mbox{``{\sf init}''}}
\newcommand{\pre}{\mbox{``{\sf pre-com}''}}
\newcommand{\com}{\mbox{``{\sf com}''}}
\newcommand{\SKIP}{\mbox{{\sf SKIP}}}
\newcommand{\lockvalue}{\mbox{{\sf lockvalue}}_q}
\newcommand{\lockround}{\mbox{{\sf lockround}}_q}
\newcommand{\clock}{\mbox{{\sf clock}}_q}
\newcommand{\status}{\mbox{{\sf status}}}
\newcommand{\Scrs}{S_\text{CRS}}
\newcommand{\Snode}{S_\text{node}}
\newcommand{\tmax}{t_\text{max}}

\title{\textbf{\textit{DEXON}}: \\A Highly Scalable, Decentralized DAG-Based Consensus Algorithm}

\author{Tai-Yuan Chen}
\author{Wei-Ning Huang}
\author{Po-Chun Kuo}
\author{Hao Chung}
\author{Tzu-Wei Chao}
\affil{DEXON Foundation, \\Taiwan}
\affil{\textit {\{popo,w,pk,haochung,n\}@dexon.org}}

\date{\today, v2.0}

\begin{document}
\maketitle

\begin{abstract}

A blockchain system is a replicated state machine that must be fault tolerant. When designing a blockchain system, there is usually a trade-off between decentralization, scalability and security. In this paper, we propose a novel blockchain system, DEXON, which achieves high scalability while remaining decentralized and robust in the real-world environment.

We have two main contributions. First, we present a highly scalable sharding framework for blockchain. This framework takes an arbitrary number of single chains and transforms them into the \textit{blocklattice} data structure, enabling \textit{high scalability} and \textit{low transaction confirmation latency} with asymptotically optimal communication overhead. Second, we propose a
single-chain protocol based on our novel verifiable random function and a new Byzantine agreement that achieves high decentralization and low latency.

\end{abstract}

\keywords{Blockchain, Blocklattice, Consensus, Byzantine Agreement, Byzantine Fault Tolerance, Replicated State Machine, Total Ordering}
\newpage
\tableofcontents
\newpage

\section{Introduction}

Blockchain systems are being challenged to demonstrate rigorous robustness and high performance in real-world situations. Many applications demand low transaction confirmation latency and high transaction throughput. However, most blockchain systems do not satisfy these criteria. For example, the confirmation latency of Ethereum is about 5 to 10 minutes and the throughput is limited to about 30 transactions per second. By contrast, some blockchain systems achieve high performance but sacrifice the robustness of the systems. For example, EOS is operated with only 21 supernodes, and is vulnerable to DDoS attacks.

We propose a novel blockchain framework, \emph{DEXON}, which achieves high performance and remains robust on the real-world Internet. To this end, we design a data structure, called \emph{blocklattice}, which allows numerous single chains to grow concurrently. Then, we propose a novel method to integrate these single chains into a globally-ordered chain without additional communication.

The blocklattice structure is a directed acyclic graph (DAG) that consists of many single chains. The blocks in the different chains acknowledge (ack) each other and collectively form this DAG structure. We use the total ordering algorithm to achieve consensus on the blocklattice so that all users are guaranteed to have the same view of the ordering of all the blocks. Because the single chains can grow concurrently, the throughput of DEXON can be easily scaled up.

We also propose our single-chain protocol, which is based on Algorand \cite{algorand} but with some minor improvements. The Algorand consensus protocol is a breakthrough that allows millions of nodes to join the protocol; each node has fair opportunity to propose and validate blocks. For this protocol, we present a new verifiable random function and Byzantine agreement achieving low confirmation time even for a large population of nodes. We emphasize that our blocklattice structure with its total ordering algorithm is a generic framework that applies to any kind of single-chain protocol. Thus, if the throughput of the underlying single-chain protocol increases by a factor of $2$, the total throughput of the whole system also increases by a factor of $2$.

\subsection{Main Features of DEXON}
DEXON has the following advantages:
\begin{description}
  \item[High Scalability] \hfill \\ Most blockchain solutions are not able to scale their throughput even with increased resources. Our blockchain system can adjust the number of chains dynamically while preserving the same latency. The throughput of our system is only bound by the available network bandwidth and computational power. In addition, the scaling methodology of the DEXON consensus is generic. That is, the blocklattice structure with its total ordering algorithm can be applied to any kind of single-chain consensus protocol.
  
  \item[Low Latency] \hfill \\ Latency is the amount of time from block proposal to confirmation. Therefore, latency is one of the paramount properties for any blockchain. We propose a fast Byzantine agreement that is expected to terminate in $6\lambda$ time, where $\lambda$ is the upper bound of the network’s gossip period. The DEXON consensus algorithm achieves second-level latency instead of traditional minute-level latency as exhibited by blockchains such as Ethereum or Bitcoin. A second-level latency blockchain opens a new era that provides numerous variations in service that cannot be delivered by traditional blockchains. 
  
  \item[High Decentralization] \hfill \\ The DEXON consensus is based on ``\emph{proof-of-participation (PoP)}''; that is, every node has equal chance to propose a block. We adopt a \emph{verifiable random function (VRF)} to decide who can issue a block; this serves to minimize the communication cost so that a large population can join the protocol.
  
  \item[Transaction Ordering Fairness] \hfill \\ In traditional blockchain systems, a single proposer can determine the transaction ordering; this makes traditional blockchain systems vulnerable to front-run attacks. By contrast, in the DEXON consensus algorithm, no single block proposer can determine the consensus timestamp of a proposed block.
  
  \item[Unpredictable Randomness] \hfill \\ Randomness is often a desired functionality in various smart contracts such as decentralized applications (DApps), particularly gaming DApps. Normally, a blockchain system cannot generate unpredictable randomness on-chain, so DApp developers must rely on some trusted third party for random input, such as the service provided by Oracalize. However, the DEXON consensus generates on-chain unpredictable randomness on the fly as it achieves consensus. Once a block has been confirmed by the DEXON Byzantine agreement, a committee of nodes generates a threshold signature, which is an unpredictable value. Thus, the hash of the signature serves as the unique, unpredictable, and unbiased randomness of the block.

  \item[Explicit Finality] \hfill \\ In blockchain systems based on proof-of-work, such as Bitcoin or Ethereum, transaction confirmation is probabilistic. Only after users have waited for a long sequence of block confirmations can a transaction be considered as \textit{probabilistically finalized}; the system is thus vulnerable to double spending attacks. For use cases such as payment networks, an explicit finality with probability 1 is necessary. In the DEXON consensus, every transaction is confirmed to be finalized with probability 1 and is secured by DEXON's Byzantine agreement algorithm, which has been proven correct with rigorous mathematical logic.
    
  \item[Low Communication Overhead] \hfill \\ A two-phase-commit consensus algorithm has an $\mathcal{O}(n^{2})$ communication complexity in any $n$-node setting and is thus highly costly to scale up. 
  The DEXON consensus adopts VRF to reduce the amount of nodes joining the protocol from $n$ parties to $\mathcal{O}(\log n)$. Thus, the communication cost is reduced from $\mathcal{O}(n^2)$ to $\mathcal{O}(n \log n)$. In this situation, the number of nodes can be scaled to millions while maintaining only hundreds of nodes that must communicate with each other.
  
  \item[Energy Efficiency] \hfill \\ DEXON Consensus has asymptotically optimal computation overhead, making it highly energy efficient.
    
  \item[Low Transaction Fees] \hfill \\ The transaction fees of a typical blockchain increase when the blockchain network is congested. The DEXON consensus is highly scalable and has low communication overhead, which enables it to maintain the lowest possible transaction fees in large-scale deployments.
\end{description}

When designing a blockchain system, a trade-off usually exists between decentralization, performance, and safety, which we call the \emph{trilemma} problem in blockchains. For example, EOS and Hashgraph achieve high performance, but they are operated by few supernodes and are vulnerable to DDoS attack. By contrast, Algorand is decentralized and has robust safety, but its throughput is limited.

In DEXON, we balance the requirements of the trilemma. DEXON has scalable transaction throughput and low confirmation latency. However, DEXON remains highly decentralized and robust in practical deployment environments.

\subsection{Related Work}
\textit{Proof-of-Work}
Bitcoin \cite{bitcoin} is the first blockchain protocol whose consensus delivers Nakamoto consensus: This means that the Bitcoin system solves a mathematical puzzle as a proof to generate next block, and once a block has been followed by six continuous blocks, the block is confirmed. This mechanism causes Bitcoin to have a latency of approximately one hour. Even more problematically, proof-of-work-based consensus consumes exorbitant quantities of energy.\vspace{1ex} \\
\textit{Proof-of-Stake}
Numerous proof-of-stake consensus systems \cite{algorand,DFINITY,HASHGRAPH,DBLP:conf/eurocrypt/DavidGKR18} have been proposed in recent years. In these schemes, the nodes with adequate stakes have the right to propose their blocks. The probability that a node can propose a block is proportional to the stakes the node owns.
\vspace{1ex} \\
\textit{DAG-based consensus}
Phantom, SPECTRE, IOTA, Conflux, and Mechcash are DAG-based consensus systems, all of which are, at their core, variants of the Nakamoto consensus. This leads two disadvantages: first, these systems demonstrate low performance (throughput is low and latency is long); second, the finality is probabilistic, allowing some attacks (such as selfish-mining) to exist. To conclude, constructing a DAG-based consensus with the Nakamoto consensus limits performance and safety. 
\vspace{1ex}\\
Algorand is a breakthrough proposed by Gilad et al. \cite{algorand}, that reduces the communication complexity from $\mathcal{O}(n^2)$ to $\mathcal{O}(n \ln n)$ and thus supports large population of nodes (e.g. 500K nodes). They use a verifiable random function (VRF) to protect nodes from DDoS attack, and the VRF is also a lottery that decides which node has the right to propose a block or to vote for each round of their Byzantine agreement protocol.
The consensus of Algorand is based on Byzantine agreement among samples from the whole set of nodes. 
Thus, the probability of the correctness of whole system is based on hypergeometric distribution.
This is the reason why Algorand can only tolerate less than one third of total number of nodes to be malicious while Algorand achieves high decentralized.
\vspace{1ex}\\
Dfinity \cite{DFINITY} is a permissioned blockchain and is designed for large population of nodes (around 10K nodes). Dfinity contains
a randomness beacon which generates new randomness by a VRF with information from new confirmed block. They use the output of a VRF to select a leader and electors for a round. 
By hypergeometric distribution, Dfinity only samples hundreds of nodes to notarize a block instead of using all nodes, and the correctness holds with high probability.
\vspace{1ex}\\
The consensus of Hashgraph \cite{HASHGRAPH} adopts Byzantine agreement on a graph and their round-based structure costs a latency of $\mathcal{O}(\ln n)$ for each round of Byzantine Agreement, which means its confirmation time increases with the number of nodes. This limits the decentralized level of Hashgraph.

\paragraph{Roadmap} In Section \ref{section:preliminary}, we introduce the cryptographic primitives and the system model used in this paper. In Section \ref{section:singleChain}, we formally introduce our Byzantine agreement protocol and describe how to build up a single chain by our Byzantine agreement protocol. In Section \ref{section:blocklattice}, we introduce how DEXON reaches a consensus on a blocklattice by using its total ordering algorithm. Finally, we discuss the sharding scheme and the method to adjust system parameters in Section \ref{section:extension}.

\section{Preliminaries and Model}
\label{section:preliminary}

\subsection{Cryptographic Primitive}
In this section, we introduce the cryptographic primitives used in this paper, including our hash function, digital signature, threshold signature, and verifiable random function.

\paragraph{Hash Function and Digital Signature} In this paper, we model the hash function $Hash$ as a random oracle. That is, the hash function acts as a truly random function which can only be accessed by ``querying'' an oracle. We also assume that an unforgeable digital signature $Sig$ is available. We define $Sig_{sk}(m)$ to be the signature of the message $m$ signed by the secret key $sk$.

\paragraph{Threshold Signature} 
Let $n$ be the number of parties. 
A \emph{$(n, t)$-threshold signature} scheme allows arbitrary $t$ parties to sign a message and any party with the public key can verify the threshold signature. 
A threshold signature scheme consists of five probabilistic polynomial-time algorithms:
\begin{itemize}
    \item KeyGen($1^\Lambda$): a distributed \emph{key generation} algorithm takes as input a security parameter $\Lambda$, and outputs a public key $PK$, a set of verification keys $\{VK_i\}_{i=1}^{n}$, and the secret key $SK_i$ for each party $i$.
    \item ShareSign($SK_i,m$): a \emph{share signing} algorithm takes as input a message $m$ and a secret key $SK_i$. It outputs a signature share $SSign_i(m)$.
    \item Verify-ShareSign($m, VK_i, \sigma_i$): a \emph{share verification} algorithm takes as input a message $m$, the verification key $VK_i$ and a signature share $\sigma_i$ on $m$ from the party $i$. It outputs $1$ if $\sigma_i$ is valid and outputs $0$ otherwise.
    \item Combine($m,PK,\{VK_i\}_{i\in T},\{\sigma_i\}_{i\in T}$): a \emph{share combining} algorithm takes as input a message $m$, the public key $PK$, the set of verification key $\{VK_i\}$ and $t$ verified signature shares $\{\sigma_i\}_{i\in T}$ on the message $m$ where $T$ is a subset of $\{1,2,\cdots,n\}$ such that $|T| = t$. It outputs a threshold signature $TSign(m)$.
    \item Verify-TSign($m,PK, \Sigma$): a \emph{signature verification} algorithm takes as input a message $m$, the public key $PK$, and a threshold signature $\Sigma$. It outputs $1$ if $\Sigma$ is valid and outputs $0$ otherwise.
\end{itemize}
It is required that except with negligible probability, Verify-ShareSign$\left(m, VK_i,\mbox{ShareSign($SK_i,m$)}\right) = 1$ holds for any party $i$. 
It is also required that except with negligible probability, Verify-TSign$\left(m,PK, \Sigma\right) = 1$ holds, where $\Sigma$ is the threshold signature of $t$ valid signature shares $\{\sigma_i\}_{i\in T}$ for any subset $T \subset \{1,2,\cdots,n\}$.

For the security definition, one can refer to \cite{Libert2014}.

\paragraph{Verifiable Random Function}
Let $(PK,SK)$ be a pair of a public key and a private key. \emph{Verifiable random function (VRF)} is a kind of pseudo-random function such that only a user that has $SK$ can compute the function whereas anyone can verify the validity of the function evaluation by $PK$ and public information. Typically, a VRF consists of the following algorithms:
\begin{itemize}
    \item KeyGen($1^\Lambda$): a probabilistic polynomial time algorithm takes as input a security parameter $\Lambda$, and generates a public key and private key pair $(PK, SK)$.
    \item VRF($SK, x$):  a deterministic algorithm takes as input the private key and an initial value $x$, and outputs verifiable random value $y$. 
    \item Prove($SK, x$): a deterministic algorithm takes as input the private key and an initial value $x$, and output the proof of correctness $w$.
    \item Verify($x, y, w$): a deterministic algorithm verifies the validity of the VRF using the proof $w$.
\end{itemize}

For the correctness and security of VRF, we refer to \cite{vrf,DBLP:conf/pkc/DodisY05}.

\subsection{Terminology and System Model}
\label{sec.model}
\paragraph{Terminology}
In DEXON, an \emph{user} is uniquely identified by its public key and secret key pair. All users can transfer and receive the stakes from other users. They can also verify the correctness of the blockchain. 
Our blockchain is maintained by a special set of users $\Snode$ where the members of $\Snode$ are called \emph{nodes}. We define two special sets of nodes: the \emph{CRS set} $\Scrs$ and the \emph{notary set} $S_\text{notary}$ (a node can be in two sets at the same time). The nodes in $\Scrs$ and $S_\text{notary}$ are crucial for our single chain algorithm, which will be explained in Section \ref{section:singleChain}.

A node can pack a batch of transactions into a \emph{block}. In our single chain algorithm, we say a node \emph{proposes} a block if the block is a candidate that can be selected by our Byzantine agreement protocol. We say a node \emph{issues} a block if the block is selected and becomes a block in the single chain.

\paragraph{System Model}
We assume the adversary can corrupt the nodes \emph{adaptively}. That is, the adversary can choose which nodes are corrupted during the protocol. The corrupted nodes are called \emph{Byzantine} and the nodes that are not corrupted are called \emph{correct}. A Byzantine node can deviate from the protocol arbitrarily; it can engage in problematic malfunctions such as sending conflicting messages, violating algorithm criteria, delaying the messages between other nodes, and so on. We also assume the adversary has full control of the network. The adversary can learn all the messages delivered on the network and determine the delay and the order of the delivered messages.

If the adversary does not delay any messages between the correct nodes, we assume the network is \emph{weakly-synchronous}. That is, there exists a known \emph{time bound} $\lambda$ for the messages between any two correct nodes. We say the network is \emph{partitioned} if the messages between the correct nodes are delayed such that the delivering time exceeds $\lambda$.

\section{Single Chain}
\label{section:singleChain}
In this section, we introduce how a set of nodes build an agreed single chain. In short, the members in $\Scrs$ generate a public randomness. Then, for each height of the single chain, the members in $S_\text{notary}$ propose their blocks. Then, they try to reach an agreement regarding who is the leader for that height and all nodes adopt the leader's proposal for the next block. As long as all the correct nodes can agree on who is the leader for each height, they can build up an agreed blockchain.
 
In Section \ref{subsection:VRF}, we introduce two useful primitives for Byzantine agreement protocol: \emph{common reference string} and \emph{verifiable random function}. In Section \ref{subsection:ba}, we introduce our Byzantine agreement protocol based on leader election. Finally, we formally describe how to build a single chain in Section \ref{subsection:merge}.

\subsection{Common Reference String and Verifiable Random Function}
\label{subsection:VRF}
In this section, we introduce two useful primitives for building a single chain: common reference string (CRS) and verifiable random function (VRF).

\paragraph{Common Reference String}
An \emph{epoch} consists of a specific number of blocks. 
The CRS in our setting is actually a public randomness generated by a deterministic algorithm for each epoch; no user in the system can predict the CRS of any future epoch.
In our setting, the CRS of epoch $i$ is updated by $R_i = Hash(TSig(R_{i-1}))$, where $TSig(\cdot)$ is a threshold signature function whose input is some set of share-signatures produced by the nodes in $\Scrs$ of epoch $i-1$.

This method is similar to DFinity \cite{DFINITY}, but their system categorizes the group when users join the system. We emphasize that the users in their system have a high incentive to be malicious because each group has non-negligible probability to have the right to propose a block and compute the randomness.
However, our system selects CRS nodes to compute the randomness according to results from a previous epoch.

\paragraph{Verifiable Random Function}
The verifiable random function introduced by Micali, Rabin and Vadhan \cite{vrf} is a type of pseudorandom function by which anyone can verify the validity of the function evaluation from public information.
Several practical VRFs have been proposed \cite{DBLP:conf/pkc/DodisY05,DBLP:conf/tcc/GoyalHKW17} and most of them are based on bilinear functions.
Algorand \cite{algorand} and Ouroboros Praos \cite{DBLP:conf/eurocrypt/DavidGKR18} have demonstrated that VRF is a powerful primitive for achieving cryptographic sortition in blockchains.
%
Our system works with PKI; each user can access all other users' public keys (preregistered in previous blocks) and user $j$ computes the verifiable random function as:
\begin{equation}\label{eq:vrf}
\left|R_i - Hash\left(Sig_{sk_j}(x)\right)\right|,
\end{equation}
where $R_i$ is CRS of epoch $i$ and $x$ is public information. 
Thus, $Hash(Sig_{sk_j}(x))$ is verifiable with $user_j$'s public key and unpredictable.  
Our VRF has three benefits compared to the VRF in Algorand.
First, our design is fairer.
The verifiable random function used in Algorand is $Hash(Sig_{sk_j}(Q_{i-1}, x))$, where $Q_{i-1}$ is the randomness from the previous block.
Whether an adversary can choose to propose a block depends on the randomness of that adversary’s block.
If the randomness is beneficial for Byzantine nodes (e.g. higher probability of proposing the next block), then the adversary proposes the block.
Thus, the overall advantage of the adversary increases up to $(1/3)/(1-1/3) = 1/2$.
The main problem is that the proposer decides the block and the randomness at the same time.
Therefore, we separate the permission of proposing a block and generating the randomness in order to avoid such bias attacks.
Second, our design is more flexible to compute because each user can compute part of the VRF for any status at any time, say, $Hash\left(Sig_{sk_j}(x)\right)$.
At the beginning of each epoch, any user can get the CRS of the epoch and compute the probability of proposing a block in the epoch.
However, Algorand's VRF requires $Q_{i-1}$ to compute $Q_i$.
Third, our design has better space consumption. In Algorand, each block must store the randomness for VRF, but our design uses the same randomness for many blocks in one epoch.
Thus, the space complexity is reduced by a constant.

A further optimization is only to update CRS when new users join. Then, the space complexity is reduced to a constant value that is independent of the number of blocks.

\subsection{Byzantine Agreement}
\label{subsection:ba}
The \emph{Byzantine general problem} was introduced by Lamport, Shostak, and Pease \cite{Lamport1982}; it allows a set of nodes to agree on a single bit $b \in \{0,1\}$ where some of the nodes may be malicious. Recently, Chen et al. \cite{Chen2018ALGORANDAS} proposed a Byzantine agreement protocol based on leader election, which achieves fast agreement in a synchronous network, upholds safety in an asynchronous network, and recovers from any partition rapidly.

Now, we introduce our Byzantine agreement protocol, which is based on Algorand's protocol but includes some minor improvements. In DEXON, each single chain is maintained by a notary set $S_\text{notary}$. Thus, only the members in $S_\text{notary}$ are involved in the Byzantine agreement protocol. Let $n$ be the number of nodes in $S_\text{notary}$ and $t$ be the number of Byzantine nodes in $S_\text{notary}$.

Comparing to Algorand, our leader election procedure is independent of the round index, so the nodes are not required to propose their values at each round. Consequently, except for the first round, the running time of each round reduces to $2\lambda$. Also, if no partition is present, our protocol terminates in $t$ rounds in the worst case and can be expected to terminate in $\frac{7}{4}$ rounds.


Let $\tmax = \lfloor (n-1)/3 \rfloor$, $R_i$ denote the CRS at epoch $i$ and $V$ denote the set of values that can be decided. We also define two special values $\bot$ and $\SKIP$ that are not in $V$. For each node $q \in S_\text{notary}$, $q$ has four internal variables: $r_q$ records the index of the round at which $q$ is working, $\lockvalue$ records the candidate value that $q$ supports, $\lockround$ records the index of the round from which $\lockvalue$ comes and $\clock$ is $q$'s local clock. Let $sk_q$ and $pk_q$ denote the secret key and public key of $q$, respectively. We define $\status$ to be the public predictable information of the block (e.g. shard ID, chain ID, block height). For the leader election, each node $q \in S_\text{notary}$ computes the signature $\sigma_q = Sig_{sk_q}(\status)$ with its secret key $sk_q$. We define three kinds of messages:
\begin{enumerate}
    \item the \emph{initial message} of the node $q$: $\left(\init,v_q,q,\sigma_q \right)$
    \item the \emph{pre-commit message} of the value $v$ from the node $q$ at the round $r$: $\left(\pre,v,q,r\right)$
    \item the \emph{commit message} of the value $v$ from the node $q$ at the round $r$: $\left(\com,v,q,r\right)$
\end{enumerate}

With these notations, we introduce the leader election algorithm which will be a subroutine of our Byzantine agreement protocol. Let $M_q$ denote the set of initial messages that the node $q$ receives from other nodes. The node $q$ verifies the signatures in $M_q$ and sets $U_q$ to be the set of nodes whose signatures are valid. Then, $q$ computes \[
\ell_q = \argmin_{j\in U_q} \left|R_i - Hash\left(\sigma_j\right)\right|.
\]
We say $\ell_q$ is the \emph{leader} of $q$.

Our Byzantine agreement protocol (Algorithm \ref{algorithm:ba}) is a round-based protocol. Initially, for all correct nodes $q \in S_\text{notary}$, $q$ initializes its internal variables by $r_q = 1$, $\lockvalue = \bot$, $\lockround = 0$ and $\clock = 0$ and also chooses its initial value $v_q \in V$. 

Our protocol has four steps in each round. At Step 1, all the nodes gossip their own initial value $v_q$ in the format $(\init,v_q,q,\sigma_q)$. 

When $\clock = 2\lambda$, $q$ enters Step 2. If $\lockvalue = \bot$, $q$ verifies the initial messages it receives and computes the set $U_q$ of nodes whose signatures are valid. If $U_q \neq \emptyset$, $q$ identifies its leader $\ell_q$ and pre-commits $\ell_q$'s value; otherwise, $q$ pre-commits $\bot$. If $\lockvalue \neq \bot$, node $q$ pre-commits $\lockvalue$. We say node $q$ pre-commits on a value $v$ if node $q$ gossips the message $(\pre,v,q,r)$ where $r$ is the round number that node $q$ is working at.

When $\clock = 4\lambda$, node $q$ enters Step 3. If node $q$ has seen $2\tmax+1$ pre-commit messages of the same value $v\in V\cup \{\bot\}$ at round $r_q$, node $q$ updates $\lockvalue = v$ and $\lockround = r_q$ and commits $v$. Otherwise, node $q$ commits $\SKIP$. Note that node $q$ must commit some value at Step 3. We say node $q$ commits on a value $v$ if node $q$ gossips the message $(\com,v,q,r)$ where $r$ is the round number that $q$ is working at. After node $q$ gossips the commit message, $q$ enters Step 4, at which $q$ waits for the forward conditions.

Suppose a node $q$ is working at round $r_q$. The node $q$ updates its internal variables as soon as one of the following conditions holds:
\begin{enumerate}
    \item If node $q$ has seen $2\tmax +1$ pre-commit messages of the same value $v \in V \cup \{\bot\}$ at the same round $r$ such that $r_q \geq r > \lockround$, $q$ sets $\lockvalue = v$ and $\lockround = r$.
    \item (forward condition) If node $q$ has seen $2\tmax +1$ pre-commit messages of the same value $v \in V \cup \{\bot\}$ at the same round $r$ such that $r > r_q$, $q$ sets $\clock = 2\lambda$, $\lockvalue = v$, $\lockround = r$ and starts the round $r$ from Step 2.
    \item (forward condition) If the node $q$ has seen $2\tmax +1$ commit messages of any value at the same round $r$ such that $r \geq r_q$, $q$ sets $\clock = 2\lambda$ and starts the round $r+1$ from Step 2. 
\end{enumerate}
We say that node $q$ achieves the \emph{forward condition}, if the condition 2 or the condition 3 holds. Node $q$ goes into the next round immediately if it achieves the forward condition even if it does not achieve the forward condition at Step 4.

Node $q$ \emph{decides} on a value $v$ as soon as node $q$ has seen $2\tmax +1$ commit messages of the same value $v \in V \cup \{\bot\}$ at the same round $r$.

The protocol for a node $q$ is summarized as Algorithm \ref{algorithm:ba}.
\begin{algorithm}[!htbp]
\caption{Synchronized Byzantine Agreement Algorithm}
\label{algorithm:ba}
\Fn{DEXON\_Byzantine\_Agreement \text{for node} q}{
    \Input{an initial value $v_q \in V$ from node $q$ and the public key $\{pk_q\}$ from all nodes}
    \Output{an agreed value $v_{fin}\in V \cup \{\bot\}$ from some node}
  Initialize $r_q = 1$, $\lockvalue = \bot$, $\lockround = 0$ and $\clock = 0$
  
  \Step { 1: when $\clock = 0$,}{
    gossip$(\init,v_q,q,\sigma_q)$
  }
  \Step { 2: when $\clock = 2\lambda$, }{
    \If{$\lockvalue = \bot$ and $U_q \neq \emptyset$}{
      node $q$ identifies its leader $\ell_q$ at $q$'s current view\\
      gossip$(\pre,v_{\ell_q},q,r_q)$
    }
    \ElseIf{$\lockvalue = \bot$ and $U_q = \emptyset$}{
      gossip$(\pre,\bot,q,r_q)$
    }
    \Else{gossip$(\pre,\lockvalue,q,r_q)$}
  }
  \Step { 3: when $\clock = 4\lambda$,}{
    \If{node $q$ has seen $2\tmax +1$ pre-commit messages of the same value $v\in V\cup \{\bot\}$ at round $r_q$}{
         $\lockvalue = v$\\
         $\lockround = r_q$\\
        gossip$(\com,v,q,r_q)$
        }
    \Else{gossip$(\com,\SKIP,q,r_q)$}
  }
  \Step { 4: when $\clock  \in (4\lambda,\infty)$}{
    wait until the \emph{forward condition} is achieved
  }
}
\end{algorithm}

\subsubsection{Agreement}
We first show that our protocol will reach agreement; that is, all the correct nodes will decide on the same value.

\begin{lemma}\label{lemma:agreement1}
Assume $t \leq \tmax$. Suppose a node $p$ receives $2\tmax +1$ commit messages of $v_p$ and another node $q$ receives $2\tmax +1$ commit messages of $v_q$. If both these $2\tmax +1$ commit messages all come from the round $r$, then $v_p = v_q$.
\end{lemma}
\begin{proof}
We prove this lemma by contradiction. Suppose $v_p \neq v_q$. Because as many as $t$ Byzantine nodes exist, there exists at least one correct node that both commits on $v_p$ and $v_q$ by the pigeonhole principle. However, correct nodes can only commit on one value at one round, which leads to a contradiction.
\end{proof}

\begin{theorem}[Agreement]
Assume $t \leq \tmax$. Regardless of partition, if a correct node $p$ decides on some value $v_p$ and a correct node $q$ decides on some value $v_q$, then $v_p = v_q$. That is, the correct nodes will never decide on different values.
\end{theorem}
\begin{proof}
Because $p$ decides on $v_p$ and $q$ decides on $v_q$, $p$ and $q$ must see $2\tmax +1$ commit messages of $v_p$ and $2\tmax +1$ commit messages of $v_q$, respectively. Suppose both these $2\tmax +1$ commit messages come from the same round $r$. By Lemma \ref{lemma:agreement1}, we have $v_p = v_q$.

Suppose the $2\tmax +1$ commit messages that $p$ receives come from the round $r_p$ and the $2\tmax +1$ commit messages that $q$ receives come from the round $r_q$. Without loss of generality, we assume $r_p < r_q$. Because there are up to $\tmax$ Byzantine nodes, there must be at least $\tmax+1$ correct nodes commit on $v_p$ so that $p$ can receive $2\tmax +1$ commit messages of $v_p$. For all rounds $r > r_p$, these $\tmax+1$ correct nodes will always pre-commit on $v_p$ until they see $2\tmax +1$ pre-commit messages of $v' \neq v_p$ at Step 3. However, only $2\tmax$ nodes remain, so these $\tmax+1$ correct nodes will never pre-commit any $v' \neq v_p$ for all $r > r_p$. Thus, for all $r > r_p$, if some value $v$ has $2\tmax +1$ pre-commit messages, then $v = v_p$. 

Because $q$ receives $2\tmax +1$ commit messages of $v_q$, there must exist at least $\tmax+1$ correct nodes that commit on $v_q$ at the round $r_q$. These $\tmax+1$ correct nodes commit on $v_q$ only if they have seen $2\tmax +1$ pre-commit messages of $v_q$ at round $r_q$. Therefore, $v_q = v_p$.
\end{proof}

\subsubsection{Termination}\label{subsubsection:baTermination}
We now analyze when the algorithm terminates if no partition exists or if the system recovers from a previous partition.

\begin{proposition}[Termination without partition]\label{proposition:terminationWOpartition}
Assume $t \leq \tmax$. If all the correct nodes start at the $r$-th round within time $\lambda$ and no partition exists, all the correct nodes will decide on some values in $t+1$ rounds. 
\end{proposition}
\begin{proof}
\hfill \break
\noindent\textit{Case 1: Some correct node has decided.}
If a correct node $p$ has decided on value $v_p$, $p$ must have seen $2\tmax +1$ commit messages of $v_p$. Because $p$ propagates these $2\tmax +1$ commit messages, all the correct nodes will hold this information after time $\lambda$ and decide on $v_p$ in one round.

\noindent\textit{Case 2: Some correct node has seen $2\tmax +1$ pre-commit messages.}
Suppose no node has decided but there exists a correct node $p$ that has seen $2\tmax +1$ pre-commit messages of a value $v_p$. Because $p$ propagates these $2\tmax +1$ pre-commit messages, all the correct nodes will hold this information after time $\lambda$. Thus, all the correct nodes will commit $v_p$ so that they can reach an agreement on $v_p$.

\noindent\textit{Case 3: No correct node has seen $2\tmax +1$ pre-commit messages.}
Because no correct node has ever seen $2\tmax +1$ pre-commit messages, $\lockvalue = \bot$ for all correct node $q$. Thus, they will identify their own leader by their local view and by the leader's pre-commit value. Because all correct nodes start at the $r$-th round within time $\lambda$, they can receive all the initial values from other correct nodes before identifying the leaders. Thus, there exists some correct nodes that pre-commit different values relative to each other only if a Byzantine node proposes different initial values to different nodes\footnote{Note that not proposing any initial value is considered to be equivalent to proposing $\bot$.}. However, the correct nodes will propagate the initial value so all correct nodes will have the same set of initial values after time $\lambda$. Thus, to prevent the correct nodes from agreeing on the same leader, Byzantine nodes must propose different initial values to different nodes at every round. However, a node can only propose an initial value once or it will be caught. Thus, the best strategy of Byzantine nodes is that different Byzantine nodes propose their initial values at different rounds so $t$ Byzantine nodes can only interfere during $t$ rounds. Thus, all the correct nodes will decide on some values in $t+1$ rounds with certainty.
\end{proof}

\begin{proposition}
Assume $t \leq \tmax$. Suppose all correct nodes start at $r$-th round within time $\lambda$ and no partition exists. Then, it is expected that all correct nodes will decide on some values in $\frac{7}{4}$ rounds.
\end{proposition}
\begin{proof}
From the proof of Proposition \ref{proposition:terminationWOpartition}, we know that if some correct node has decided on a value $v$ or has seen $2\tmax +1$ pre-commit messages of a value $v$, then all the correct nodes will decide on $v$ in one round.

In a network without partition, the best strategy for the Byzantine nodes has been described in Case 3 in the proof of Proposition \ref{proposition:terminationWOpartition}. However, to interfere with $k$ rounds successfully, the Byzantine nodes must win the leadership \footnote{Notice that the leader is the node $j$ with the lowest value $\left|R_i - Hash\left(\sigma_j\right)\right|$.} in the following $k$ rounds. The probability of such an event is \[
\prod_{i=0}^{k-1} \left(\frac{t-i}{n-i}\right) \leq \left(\frac{t}{n}\right)^k.
\]
Thus, in expectation, the number of interfered rounds can be upper-bound by \[
\sum_{i=1}^t \left(\frac{t}{n}\right)^i \cdot i \leq \frac{\frac{t}{n}}{(1-\frac{t}{n})^2}.
\]
Because $n \geq 3\tmax+1 \geq 3t+1$, the expected number of interfered rounds is less than $\frac{3}{4}$. After an interfered round, all the correct nodes will decide in the next round, so it is to be expected that they will all terminate in $\frac{7}{4}$ rounds.
\end{proof}

\begin{proposition}[Termination when the partition is resolved]
Assume $t \leq \tmax$. If the partition is resolved, all the correct nodes will decide on some values in $t+2$ round. It is to be expected that all correct nodes will decide on some values in $\frac{11}{4}$ rounds.
\end{proposition}
\begin{proof}
If there exists a node $p$ that has decided on a value $v_p$, $p$ must have seen $2\tmax +1$ commit messages of $v_p$. All the correct nodes will receive these $2\tmax +1$ commit messages of $v_p$ within time $\lambda$ after the partition is resolved and decide on $v_p$.

Suppose no node has decided and $p$ is the node working on the latest round $r_p$. To enter the round $r_p$, $p$ must achieve the forward condition at round $r_p - 1$. Because the partition is resolved, all correct nodes will also achieve the forward condition within time $\lambda$ after the partition is resolved and also enter the round $r_p$. Later on, if some node $q$ achieves the forward condition and enters the round $r_p + 1$, other correct nodes will also achieve the forward condition within time $\lambda$. Thus, all correct nodes start at the round $r_p + 1$ with time difference $< \lambda$ and Proposition \ref{proposition:terminationWOpartition} guarantees that they will decide on some values within the following $t+1$ rounds. Similarly, it is to be expected that all correct nodes will decide on correct values in $\frac{11}{4}$ rounds.
\end{proof}

\subsubsection{Communication Complexity}
We now analyze the communication complexity of a single node for a single round. Because a correct node will help to propagate the messages, all correct nodes will gossip $\mathcal{O}(n)$ messages in a single round. Thus, the communication complexity for all nodes is $\mathcal{O}(n^2)$ in a single round.

As discussed in Section \ref{subsubsection:baTermination}, if no partition exists or the system recovers from a partition, our Byzantine agreement protocol terminates in $t+1$ rounds in the worst case and is expected to terminate in $\frac{11}{4}$ rounds. We assume $n \geq 3t+1$, so the protocol terminates in $\mathcal{O}(n)$ rounds in the worst case and is expected to terminate in $\mathcal{O}(1)$ rounds. Therefore, the total communication complexity of the protocol is $\mathcal{O}(n^3)$ in the worst case and $\mathcal{O}(n^2)$ in the expected case.

\subsection{Consensus on a Single Chain}
\label{subsection:merge}
In this section, we describe how to build up a single chain by Byzantine agreement protocol. Our single chain protocol has four parts: system setup, registering public key, proposing block, notarizing blocks and updating epoch.

\paragraph{Overview}
The members in $\Scrs$ are responsible for updating the CRS. The members in $S_\text{notary}$ have right to propose a block and are responsible for running the Byzantine agreement protocol to decide whose block can be chosen for the single chain.

When the members in $S_\text{notary}$ reach an agreement of a block, they sign their respective signature shares of the agreement. The block proposer of the next block has to verify these signature shares and combine them into a threshold signature $\Sigma_h$. The threshold signature $\Sigma_h$ should be stored in the next block. If someone tampers the content of a block, the hash of the block changes and become inconsistent with the threshold signature in the next block. Because it is infeasible to forge the threshold signature, the integrity of the blocks is guaranteed.

The members in $\Scrs$ and $S_\text{notary}$ are re-elected for each epoch. Each epoch $i$ corresponds to a randomness $R_i$ of CRS. When $\Scrs$ is re-elected, the members in $\Scrs$ will update the CRS for the next epoch. 

\paragraph{Registering Public Key}
Any user with enough deposit can register its own public key in the chain. The users who want to enter or leave $\Snode$ can announce the request. The request will be recorded in the blocks. The node is called \emph{valid} if it owns enough deposit\footnote{To avoid nothing-at-stake problem, the node who wants to register its public key needs to lock its deposit. In practice, this can be done by smart contracts.} and has registered a valid public key. Let $view^i$ denote the snapshot of the whole system at the first block of the epoch $i$. Let $\Snode^i$ denote the set of all valid nodes in $view^i$. The set $\Snode^{i+1}$ can be updated by $\Snode^i$ and the blocks in the epoch $i$.

\paragraph{Proposing Blocks}
Only the nodes in $S_\text{notary}$ can propose new blocks. To propose a block, a node $q$ computes the hash $v_q$ of the block and takes $v_q$ as the initial value for the Byzantine agreement protocol. Note that node $q$ can send the hash of the block to join the BA scheme as early as possible, but has to propose the complete block immediately so that other nodes can determine the block and commit it.

A node is allowed not to propose any block. Because the leader election is determined by the public predictable information $\status, R_i$, and the secret keys, the node $q$ can compute the value $\left|R_i - Hash\left(Sig_{sk_q}(\status)\right)\right|$ locally. If the value is too large, $q$ is unlikely to be chosen as the leader so it may give up the chance to propose the block.

Practically, we can set a threshold $\delta$ of proposing a block. That is, the node $q$ will propose a block if and only if $\left|R_i - Hash\left(Sig_{sk_q}(\status)\right)\right| \leq \delta$. Suppose $n$ is the number of nodes in $S_\text{notary}$. Then, the probability that at least one node proposes a block is $1-(1-\delta)^n$. Choosing $\delta = 20/n$, the probability is larger than $1-10^{-8}$. If no one proposes a block, all the correct nodes will pre-commit $\bot$ and they will soon decide on $\bot$. Thus, the protocol will not halt. All the correct nodes can agree on an empty block and start the protocol for the next block soon.

\paragraph{Notarizing Blocks}
The nodes in $S_\text{notary}$ are responsible for running the Byzantine agreement protocol (Algorithm \ref{algorithm:ba}) to decide the next block of the chain. As shown in Section \ref{subsection:ba}, as long as there is no partition, the members in $S_\text{notary}$ can reach an agreement even if they do not start within time $\lambda$.

To notarize a block, we choose a $(n,\tmax+1)$-threshold signature scheme. When a node $q \in S_\text{notary}$ decides on some value $v$ for the block with height $h$, node $q$ signs a signature share $\sigma_q$ by its secret key $SK_q$\footnote{Note that the key pair $(PK_q,SK_q)$ for the threshold signature is different from the key pair $(pk_q,sk_q)$ that the node $q$ registers on the block.} and announces $\sigma_q$. The block proposers of the next block (of height $h+1$) have to verify at least $\tmax+1$ signature shares for the block with height $h$ and combine the signature shares into a threshold signature $\Sigma_h$. The threshold signature $\Sigma_h$ should be packed into the next block. We say a block with height $h$ is \emph{notarized} if its threshold signature $\Sigma_h$ is packed in the next block and the next block is issued.

Assume the number of Byzantine nodes is up to $\tmax$. Then, anyone can make sure the block with the hash $v$ is notarized by $S_\text{notary}$ if more than $\tmax+1$ nodes in $S_\text{notary}$ have signed on $v$. To verify the fact, we only need to check whether the block hash $v$ matches the threshold signature $\Sigma_h$ stored in the next block and verify whether $\Sigma_h$ is valid or not.

\paragraph{Updating Epoch}
The members in $\Scrs$ and $S_\text{notary}$ are re-elected for each epoch. Let $\Scrs ^i$ and $S_\text{notary}^i$ denote the CRS set and the notary set of the epoch $i$.

The set $\Scrs ^i$, the set $S_\text{notary}^i$, and the randomness $R_i$ are decided before the epoch $i$ starts. When the epoch $i$ starts, the members in $\Scrs ^i$ generate the randomness $R_{i+1}$ by \[
R_{i+1} = Hash\left(TSig(R_i)\right),
\]
where $TSig(R_i)$ is the threshold signature signed by the nodes in $\Scrs ^i$. After $R_{i+1}$ is decided, the set $\Scrs ^{i+1}$ and the set $S_\text{notary}^{i+1}$ can be elected by Fisher-Yate shuffle with the randomness $R_{i+1}$ from the set $\Snode^i$. Finally, the members in $\Scrs ^{i+1}$ and $S_\text{notary}^{i+1}$ run the key generation algorithm of the threshold signature scheme, respectively. 

The timeline of building a single chain is summarized in the Figure \ref{figure:singlechain}.
\begin{figure}[htpb]
    \centering    
    \includegraphics[width=16cm]{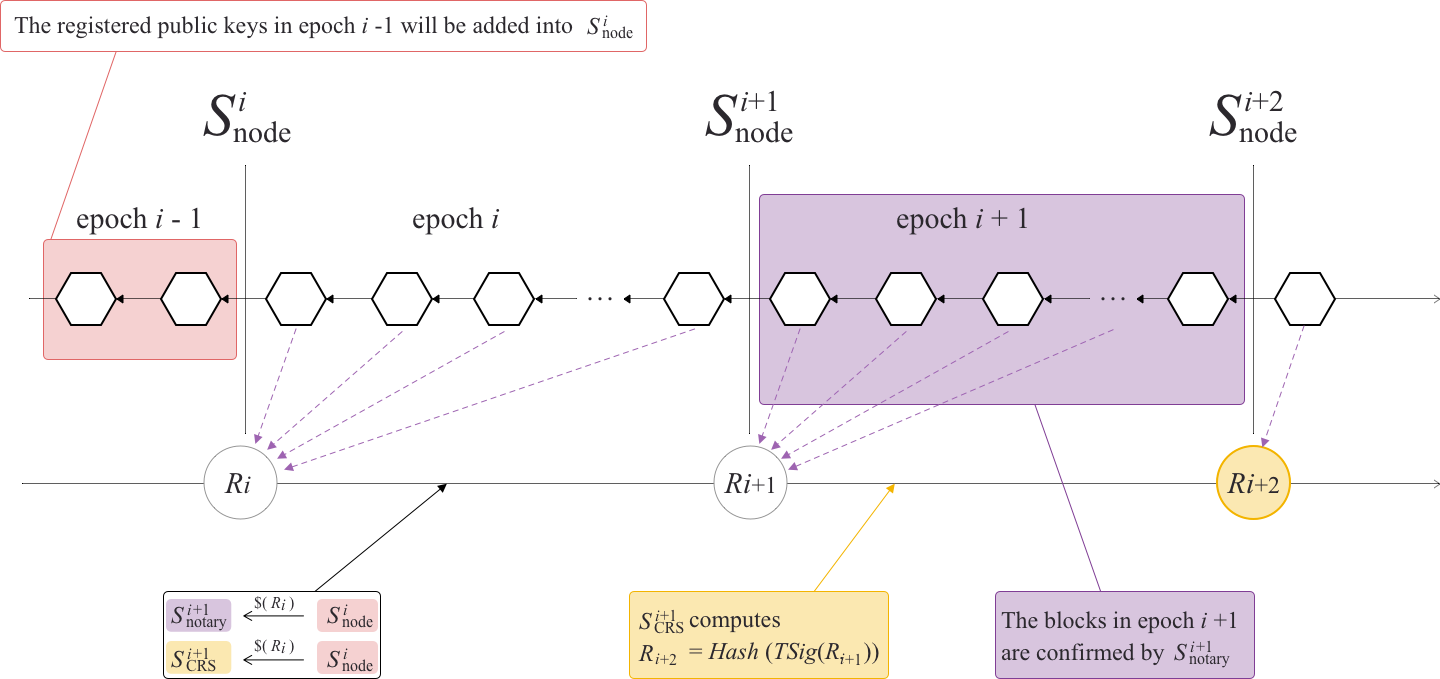}
    \caption{The consensus of a single chain}
    \label{figure:singlechain}
\end{figure}

\section{Blocklattice}
\label{section:blocklattice}

In DEXON, the blocklattice arises from numerous single chains and the ack information in the blocks of those chains.
We introduce the total ordering algorithm on the blocklattice to produce a globally-ordered chain, called the \emph{compaction chain}.

In this section, we first present our main algorithm to compact the proposed
 blocklattice into the compaction chain and then to 
generate the timestamp for each block in the compaction chain.
For the compaction, we introduce a basic total ordering algorithm in Section \ref{subsec:toto}.
In Section \ref{subsec:timestamping}, we present our timestamping algorithm that ensures the unbiased consensus timestamp of blocks in the blocklattice, which is close to real-world time. We present the consensus on a blocklattice in Section \ref{subsec:blocklattice}.

\subsection{Notation}
\label{subsec:totalNotation}

Let $\mathcal{N}$ be the set of chains and $\mathcal{F}$ be the set of chains which contains the blocks from Byzantine nodes in 
 the DAG, whose maximum capacity is limited by $f_{max} = \lfloor (|\mathcal{N}|-1)/3 \rfloor$.
Let $\mathcal{P}$ denote the set of pending blocks, which means the block is received by a node and has not been output into the compaction chain.

 When a block $A$ follows another block $B$, we call this relation $A$ \textit{acks} $B$.
Let $\mathcal{C}$ be the set of candidate blocks that only ack the blocks that are already in the compaction chain.
 Informally, $\mathcal{A}$ denotes the set of preceding
 candidate blocks that have higher priority than
 other candidate blocks, which corresponds to the ``source message'' in the TOTO protocol \cite{DolevKM93}.

 We use the subscript $(\cdot)_p$ to denote a set or 
 function in $p$'s local view, for example, $\mathcal{C}_p$
 is the candidate set in $p$'s local view.
 We use $b_{q,i}$ to denote the $i$-th block 
 from node $q$.
 %
We define \textit{indirect ack} by the transitive law of ack: if a block $A$ acks another block $B$ through the transitive law of ack, then we say $A$ indirectly acks $B$. For example, if $A$ acks $B$ and $B$ acks $C$, $A$ indirectly acks $C$.
An example is provided in Figure \ref{figure:example}.
 \begin{figure}[htpb]
     \centering    \includegraphics[width=\textwidth,height=6cm,keepaspectratio=true]{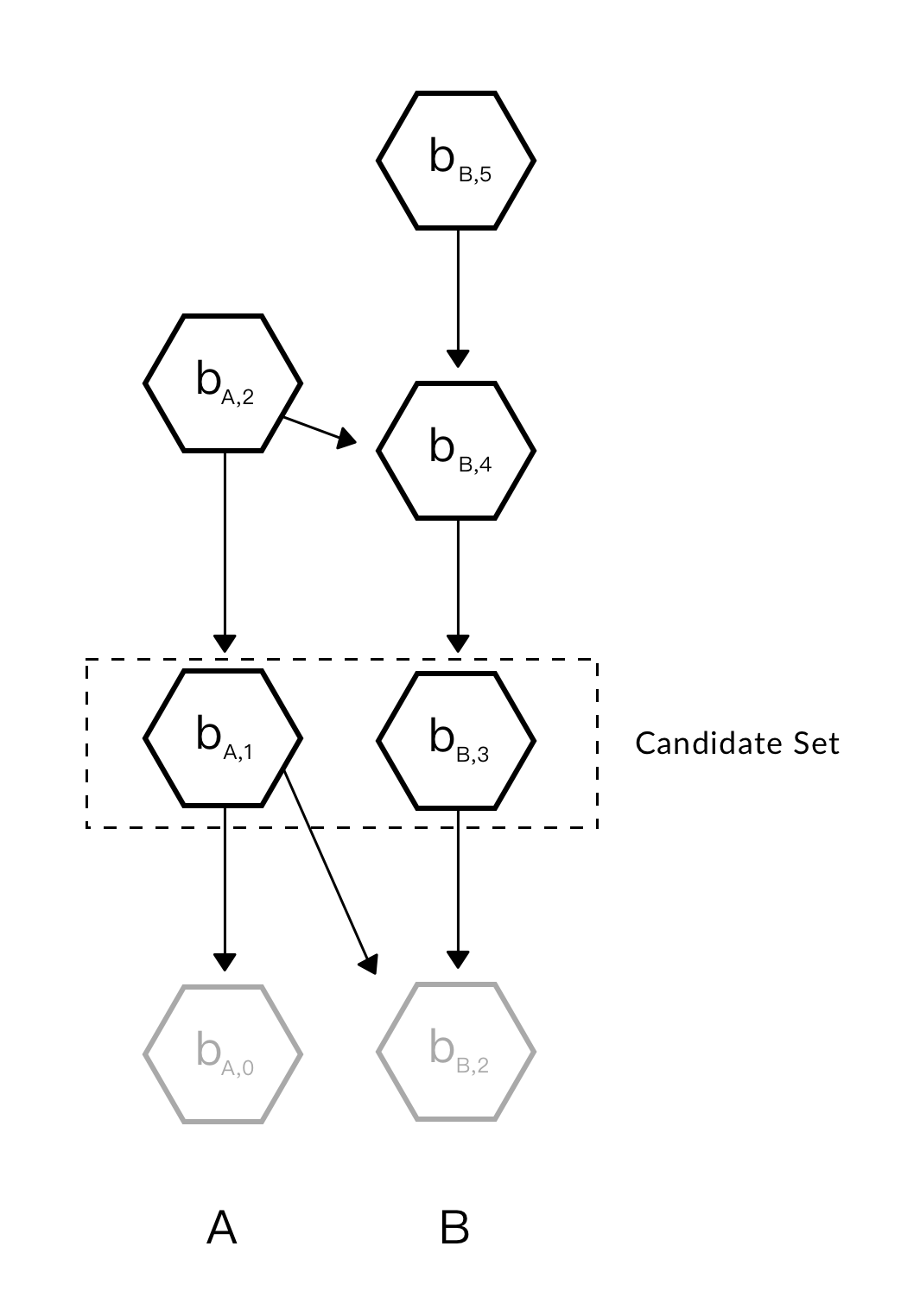}
     \caption{A simple example illustrates the notation: $\mathcal{N} = \{A, B \}$,
 $\mathcal{C} = \{b_{A,1}, b_{B,3}\}$. 
 Also, $b_{A,2}$ acks $b_{A,1}$ and $b_{B,4}$. Both $b_{A,2}$ and $b_{B,5}$ indirectly ack $b_{B,3}$.}
     \label{figure:example}
 \end{figure}

Next, we present the definition of specifications of our total ordering algorithm as the following:
\begin{enumerate}
    \item \textit{Correctness}: all non-Byzantine nodes will eventually generate the same ordered chain.
    \item \textit{Liveness}: the system will not halt with up to $f_{max}$ Byzantine nodes, and every block will eventually be finalized.
\end{enumerate}

For a formal definition, we follow the specifications in \cite{total-order-broadcast-survey}.

%
\subsection{DEXON Total Ordering Algorithm}
\label{subsec:toto}
First, we introduce the DEXON total ordering algorithm, 
which ensures that the blocklattice data structure can be compacted 
into the compaction chain.
The DEXON total ordering algorithm is described with a parameter $|\mathcal{N}|\geq\Phi>|\mathcal{N}|/2$;
it can be seen as the threshold of the output criterion.
We set $\Phi = 2f_{max} + 1$ to guarantee the liveness of the total ordering algorithm.
In this scenario, a problem arises: each node receives all block information asynchronously, thus, the order of blocks from the blocklattice cannot be decided directly by the receiving times.
Therefore, we introduce the DEXON total ordering algorithm; it 
is a symmetric algorithm and it outputs the total order for each valid block.
The DEXON total ordering algorithm is based on \cite{DolevKM93}, 
which is a weak total order algorithm \cite{total-order-broadcast-survey}. The main idea of 
the DEXON total ordering algorithm is to dynamically maintain a 
directed acyclic graph (DAG) from the received blocks. 
More precisely, each vertex corresponds to some block and 
each edge corresponds to some ack relation between blocks.
Intuitively, once a block in the graph obtains enough acks from other blocks,
the algorithm outputs the block.

Before we state the algorithm, we will first introduce some 
functions and properties for the algorithm.
We use the following three functions as potential 
functions that evaluate the quality
 of each candidate block in order to decide the output order.
\begin{enumerate}
\item  \textit{Acking Node Set}, $ANS_p(b)= \Set{j : \exists b_{j,k} \in \mathcal{P} \text{ s.t. } b_{j,k} \text{ directly or indirectly acks } b} \cup \{ p \}$, 
is the set of nodes that issued blocks ack block $b$ in node $p$'s local view. Moreover, the global Acking Node Set $ANS_p$ for node $p$ is defined by $\bigcup_{b\in \mathcal{C}_p} ANS_p(b)$.\\
\item $AHV_p(b)$ is the \textit{Acking block-Height Vector} of block $b$ in node $p$'s view, defined as:\\
\begin{center}
  $AHV_p(b)[q]=\begin{cases}
    \bot, & \text{if $q \notin  ANS_p$}\\
    k,    & \text{if $b_{q,k}$ acks $b$, 
             where $k = \min\{i : b_{q,i}\in \mathcal{P}\}$} \\
    \infty, & \text{otherwise}.
  \end{cases}$
\end{center}

\item $\#AHV_p= |\Set{j:AHV_p(b)[j]\neq \bot \text{ and } AHV_p(b)[j]\neq \infty } |$
is the number of integer elements in $AHV_p(b)$.
%
\end{enumerate}
Next, we define two functions,
\begin{center}
  $Precede_p(b_1,b_2)=\begin{cases}
    1, & \text{if $|\{j : AHV_p(b_1)[j]<AHV_p(b_2)[j]\} |>\Phi$}\\
    -1, & \text{if $|\{j : AHV_p(b_1)[j]>AHV_p(b_2)[j]\} |>\Phi$}\\
    0, & \text{otherwise},
  \end{cases}$ \\ 
\end{center}
  and 
\begin{center}
  $Grade_p(b_1,b_2) =\begin{cases}
    1, & \text{if $Precede_{p+}(b_1,b_2)=1$}\\
    0, & \text{if } |\{j : AHV_{p+}(b_1)[j]<AHV_{p+}(b_2)[j]\} | < \Phi-(n-|ANS_{p}|)\\
    \bot, & \text{otherwise},
  \end{cases}$ \\ 
\end{center}
where $(\cdot)_{p+}$ is the local view of $p$ in the future.
That is, the pending set becomes $\mathcal{P}_p \cup \mathcal{S}_p$, where $\mathcal{S}_p$ consists of all the probable blocks received by $p$ in the future.
The $Grade_p(b_1,b_2)$ function outputs the three possible relations 
between block $b_1$ and $b_2$ in the future: the first possibility is $Grade_p(b_1,b_2) = 1$, which means block $b_1$ always precedes 
block $b_2$ regardless of arbitrary following input. 
The second possibility is $Grade_p(b_1,b_2) = 0$ which means block $b_1$ cannot precede
block $b_2$ regardless of arbitrary following input.
The last possibility comprises all other relations.

We say a candidate block $b$ is \textit{preceding} if 
$\forall$ $b' \in \mathcal{C}_p$,  $Grade_p(b',b)=0$, 
and the \textit{preceding set} $\mathcal{A}$ is the set of all preceding blocks, which means that such blocks have a relatively high priority to be output by the algorithm.
Thus, 
regardless of which blocks are received afterwards, the blocks in $\mathcal{C} \setminus \mathcal{A}$
cannot precede the blocks in $\mathcal{A}$.\\
Now, we present the DEXON total ordering algorithm, which is a symmetric algorithm. 
The DEXON total ordering algorithm is an event-driven online algorithm. 
The input is one block per iteration and the algorithm produces blocks when specific criteria are satisfied. 
Regardless of the order of the inputs, the algorithm is executed by each node individually, and the algorithm outputs blocks in the same order.
The only requirement  regarding input is that the set (including acking information) of input for each node is the same eventually.  \\
The criteria consist  of two parts: the \textit{internal stability} and the \textit{external stability}. 
Informally, 
internal stability ensures each block in the preceding set has the highest priority to be output compared with all other blocks in the candidate set,
 and external stability ensures that each block in the preceding set always has
 higher priority to be output irrespective of the type of blocks received.
Let $n$ be the number of total nodes. The \emph{criteria} to output the preceding set $\mathcal{A}_p$ for node $p$ is defined as follows:
\begin{itemize}
\item Internal Stability:
  $\forall b \in \mathcal{C}_p\setminus\mathcal{A}_p, \exists b' \in \mathcal{A}_p$ s.t. $Grade_p(b',b) = 1$
\item External Stability:
\begin{itemize}  
\item (a) $|ANS_p| = n$, or
\item (b) $\exists b \in \mathcal{A}_p$ s.t. $\#AHV_p(b) > \Phi$ and \\
\phantom{lllll} $\forall b \in \mathcal{A}_p$, $|ANS_p(b)| \geq n-\Phi$
\end{itemize}
\end{itemize}

We call the output \textit{normal delivery} if it satisfies internal stability and external stability (a), and \textit{early delivery} if it satisfies internal stability and external stability (b).

Next, we simplify two \emph{criteria} by the following theorems.

\begin{theorem} [Simplifying Criterion for normal delivery]
External Stability (a) implies Internal Stability. That is,
if $|ANS_p| = n$, then, $\forall b \in \mathcal{C}_p\setminus\mathcal{A}_p, \exists b' \in \mathcal{A}_p$ s.t. $Grade_p(b',b) = 1$.
\label{thm.sim1}
\end{theorem}
\begin{proof}
First, we prove $\mathcal{A}_p \neq \emptyset$ by contradiction. 
Assume $\mathcal{A}_p = \emptyset$.
By the definition of a preceding set,
$\forall b \in \mathcal{C}_p, \exists b' \in \mathcal{C}_p$ s.t. 
$Precede_p(b',b)=1$. However, by the definition of $Precede_p$ function, $b' \in \mathcal{A}_p$, contradiction.\\
Second, assume $\exists b' \in \mathcal{C}_p\setminus \mathcal{A}_p, \forall b \in \mathcal{A}_p$ s.t. $Precede(b,b') \neq 1$.
Case 1: $\forall b'' \in \mathcal{C}_p\setminus \mathcal{A}_p$, $Precede(b'',b') \neq 1$, then $b' \in \mathcal{A}_p$, contradiction.
Case 2: $\exists b'' \in \mathcal{C}_p\setminus \mathcal{A}_p$, $Precede(b'',b') = 1$, then $b'' \in \mathcal{A}_p$, contradiction.\\
Third, if $\mathcal{A}_p = \mathcal{C}_p$, it is easy to check the correctness for the theorem. 
\end{proof}

\begin{theorem} [Simplifying Criterion for early delivery]
If both internal stability and $\mathcal{C}_p\setminus\mathcal{A}_p \neq \emptyset$ hold, then, external stability (b) also holds.
\label{thm.sim2}
\end{theorem}
\begin{proof}
It is easy to prove $\mathcal{A} \neq \emptyset$. Let $b' \in \mathcal{A}$ be the block that precedes a block, then $\#AHV_p(b') > \Phi$ holds.
Next, assume $\exists b \in \mathcal{A}$ s.t. $|ANS_p(b)| < n-\Phi$. Then, $\#AHV_p(b) < n-\Phi$ holds.
$|\{j|AHV_p(b')[j]<AHV_p(b)[j]\} |>\Phi - (n-\Phi) = 2\Phi-n$ and for the case $|ANS_p|>\Phi$. 
When the blocks issued by all nodes have been received, $|\{j|AHV_p(b')[j]<AHV_p(b)[j]\} |$ can reach more than $(2\Phi-n)+(n-\Phi) = \Phi$. This means $b$ is preceded by $b'$.
Thus, $b \notin \mathcal{A}$, contradiction.
Therefore,  $\forall b \in \mathcal{A}$ s.t. $|ANS_p(b)| \geq n-\Phi$.
\end{proof}

Hence, the $criteria$ are simplified into the following conditions:

\begin{itemize}
\item Normal Delivery: 
    $|ANS_p| = n$
\item Early Delivery: 
    $\forall b \in \mathcal{C}_p\setminus\mathcal{A}_p, \exists b' \in \mathcal{A}_p$ s.t. $Grade_p(b',b) = 1$ and $\mathcal{C}_p\setminus\mathcal{A}_p \neq \emptyset$
\end{itemize}

\begin{algorithm}[htbp]
\caption{DEXON Total Ordering Algorithm}
\label{algorithm:toto}
\Fn{DEXON\_Total\_Ordering \text{for node} p}{
    \Input{a block $b_{q,i}$ from node $q$ and its acking information per iteration}
    \Output{ordered block series}
\When( receiving a block $b_{q,i}$, ){}{ 
 \If{$b_{q,i}$ only acks the blocks that have been output,}{
        $\mathcal{C}_p = \mathcal{C}_p \cup \{ b_{q,i}\}$ \\
        \ForEach{$ r \in ANS_p$}{
            $AHV_p(b_{q,i})[r]=\infty$} 
        }
 \If{$q \notin  ANS_p$,}{
        \ForEach{ $b \in \mathcal{C}_p$ }{ 
            \If{$b$ is directly or indirectky acked by $b_{q,i}$,}{
            $AHV_p(b)[q]=i$, \\
            $ANS_p(b)=ANS_p(b)\cup \{ q \}$
            }
        \Else{
            $AHV_p(b)[q]=\infty$}
         }
        }
    }
\When( $criteria_p$ holds, ){}{
        output $\mathcal{A}_p$ in the lexicographical order of the hash value of each block\\
        //update $\mathcal{A}_p$ and $\mathcal{C}_p$\\
        $\mathcal{C}_p = \mathcal{C}_p \setminus \mathcal{A}$ \\
        $\mathcal{A}_p \leftarrow \emptyset$   \\
        \ForEach{$b \in \mathcal{P}_p$ only acks the blocks that have been output}{
          $\mathcal{C}_p = \mathcal{C}_p \cup \{ b\}$ 
          }
        \ForEach{$ b \in \mathcal{C}_p$}{ 
            compute $AHV_q(b), ANS_q(b)$}
        \ForEach{$ b \in \mathcal{C}_p$}{ 
            \If{$b$ is preceding,}{
                $\mathcal{A}_p = \mathcal{A}_p \cup \{ b \}$
                }
            }
    }
}
\end{algorithm}
The DEXON total ordering algorithm is listed as Algorithm \ref{algorithm:toto}.
There are two events that must be considered: one is a block received and the other is that some criteria are satisfied.
When a block $b_{q,i}$ is received, the algorithm updates its potential function
of candidate blocks according to the ack information from $b_{q,i}$.
If $b_{q,i}$ only acks the blocks that have been output, it is a candidate block.
Then, the algorithm updates $AHV_p(b_{q,i})$ according to $ANS_p$.
Otherwise, $b_{q,i}$ is not a candidate. If node $q$ has never been in
node $p$'s view, each candidate block updates its potential function.
If node $q$ has been in node $p$'s view, there must exist a $j < i$ s.t.
$b_{q,j} \in \mathcal{P}$. 
Thus, both potential functions $AHV_p, ANS_p$ for all candidate blocks
would not change and the algorithm would do nothing.
When the second event occurs, this means that the potential functions of the 
preceding candidate blocks are adequate and the preceding candidate can be
output.
After the algorithm outputs, it continues to collect the next candidate blocks
and to update their potential functions. 
\begin{figure}[htpb]
    \centering    \includegraphics[width=\textwidth,height=6cm,keepaspectratio=true]{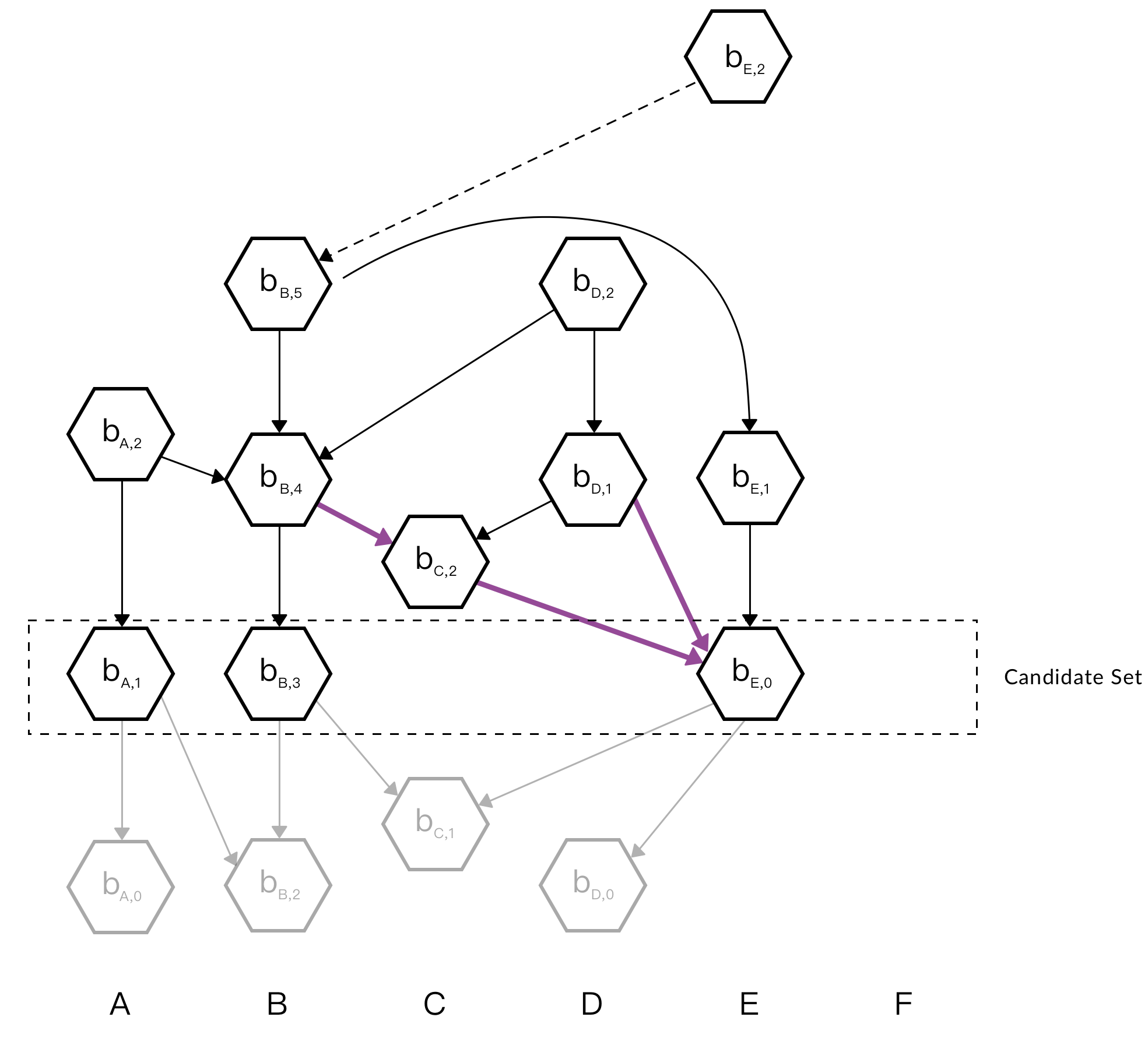}
    \caption{Example of the DEXON total ordering algorithm: 
if this is in node $C$'s local view,
$\mathcal{C} = \{b_{A,1}, b_{B,3}, b_{E,0}\}$,
$\mathcal{A} = \{b_{A,1}, b_{B,3}, b_{E,0}\}$,
$AHV_C(b_{A,1}) = (1, \infty,\infty,\infty,\infty,\bot)$,
$AHV_C(b_{B,3}) = (\infty,3,\infty,\infty,\infty,\bot)$,
$AHV_C(b_{E,0}) = (\infty,\infty,2,1,0,\bot)$,
$\#AHV_C(b_{A,1}) = 1$,
$\#AHV_C(b_{B,3}) = 1$,
$\#AHV_C(b_{E,0}) = 3$,
$ANS_C(b_{A,1}) = \{ A\}$,
$ANS_C(b_{B,3}) = \{ A,B,D\}$,
$ANS_C(b_{E,0}) = \{ A,B,C,D,E\}$
}
    \label{figure:total}
\end{figure}

\subsubsection{Correctness and Liveness}
%
\begin{lemma} [Consistency of AHV]
The first time the criterion holds for node $p,q$,
$AHV_p(b)$ and $AHV_q(b)$ are consistent for any block $b \in \mathcal{C}_p \cap \mathcal{C}_q$.
That is, for each node $r$, if $AHV_p(b)[r] \neq \bot$ and $ AHV_q(b)[r] \neq \bot$, then
$AHV_p(b)[r]=AHV_q(b)[r]$.
\label{thm.ahv}
\end{lemma}
\begin{proof}
Let $r \in \mathcal{N}$ satisfy $AHV_p(b)[r] \neq \bot$ and $ AHV_q(b)[r] \neq \bot$.
Consider that if both functions $AHV_p(b)[r] = AHV_p(b)[r] = \infty$ hold, then this lemma is correct.
Without loss of generality, assume $AHV_p(b)[r] = k $, $AHV_q(b)[r] = k'$ or $\infty$, for some integer $k' \neq k$, because node $r \in \mathcal{N}_q$.
Case 1: $AHV_q(b)[r] = \infty$, let $i$ be the minimum number s.t. $b_{r,i} \in \mathcal{P}_q$. If $i<k$, $AHV_p(b)[r] < k $, which leads to a contradiction. If $i>k$, $b_{r,k}$ must be in $\mathcal{P}_q$.
By the causality of blocks in the blocklattice, $i$ must be equivalent to $k$, which leads to a contradiction.
Case 2: $AHV_q(b)[r] = k'$, the argument is similar to case 1, thus, by the causality of blocks in the blocklattice, $k'=k$, which leads to a contradiction.
\end{proof}

\begin{lemma} [Consistency of Grade]
The first time the criterion holds for node $p,q$,
$Grade_p(b_1,b_2)$ and $Grade_q(b_1,b_2)$ are consistent for any two blocks $b_1, b_2 \in \mathcal{C}_p \cap \mathcal{C}_q$.
That is, $Grade_p(b_1,b_2)=Grade_q(b_1,b_2)$.
\label{thm.grade}
\end{lemma}
\begin{proof}
Because both of the local views of nodes $p$ and $q$ are partial DAGs, they will be the same eventually. By Lemma \ref{thm.ahv}, the output of function $AHV(b)$ of node $p$ and $q$ are consistent for every block $b \in \mathcal{C}_p \cap \mathcal{C}_q$. Thus, $Grade_p(b_1,b_2)=Grade_q(b_1,b_2)$.
\end{proof}

\begin{lemma} [Arrival of the Preceding set]
When the first time the criterion holds for node $p,q$, then $\mathcal{A}_p \subset \mathcal{C}_q$ and $\mathcal{A}_q \subset \mathcal{C}_p$ holds.
\label{thm.arrival}
\end{lemma}
\begin{proof}
We have three cases: first, node $p,q$ are normal delivery; second, one node is normal delivery and the other is early delivery; third, both are early delivery.
Case 1: we have $|ANS_p|=n$, so $\mathcal{A}_q \subset \mathcal{C}_q \subset \mathcal{C}_p$ and $|ANS_q|=n$. Consequently, $\mathcal{A}_p \subset \mathcal{C}_p \subset \mathcal{C}_q$. The lemma holds in this case.

Case 2: Without loss of generality, we assume $|ANS_q|=n$ and early deliver happened in node $p$.
First, $\mathcal{A}_p \subset \mathcal{C}_p \subset \mathcal{C}_q$ because $|ANS_q|=n$.
Second, we prove $\mathcal{A}_q \subset \mathcal{A}_p$ by contradiction.
Assume there exists a block $b \in \mathcal{A}_q\setminus\mathcal{A}_p$.
By the definition of the preceding set, $\exists b' \in \mathcal{A}_p$, s.t. $Precede_p(b',b)=1$.
However, because $b \in \mathcal{A}_q, \forall b'' \in \mathcal{A}_q, Precede_q(b'',b)=0$.
By the Lemma \ref{thm.ahv}, $b' \in \mathcal{C}_q$ is also present in the preceding set, and  $Precede_q(b',b)=1$, which leads to a contradiction.
This means $\mathcal{A}_q \subset \mathcal{A}_p \subset \mathcal{C}_p$.

Case 3: $\forall b \in \mathcal{A}_p$, let $b' \in \mathcal{A}_q$ be the block with $\#AHV_q(b')>\Phi$. Because $|ANS_q(b')|>\Phi$ and $|ANS_p(b)| \geq n-\Phi$, there must exist a node $r \in ANS_p(b) \cap ANS_q(b')$ by the pigeonhole principle. 
Thus, $b$ and $b'$ are received if the blocks issued by $r$ are received; this means $b \in \mathcal{C}_q$.
Furthermore, $\mathcal{A}_q \subset \mathcal{C}_p$ holds as well.
\end{proof}

\begin{theorem} [General Correctness for total ordering]
Let $\mathcal{A}_p^i, \mathcal{A}_q^i$ be the $i$-th set output by node $p,q$,
then $\mathcal{A}_p^i = \mathcal{A}_q^i$ holds.
\label{thm.output}
\end{theorem}
\begin{proof}
We use inductive assumption to prove this theorem.
For $i=1$ (the first time the criterion holds), for node $p,q$, we prove that $\mathcal{A}_p^i = \mathcal{A}_q^i$:
we prove this by contradiction. 
Assume $\mathcal{A}_p \neq \mathcal{A}_q$.
Without loss of generality, there exists a block $b \in \mathcal{A}_p \setminus \mathcal{A}_q$.
By Lemma  \ref{thm.arrival}, we have $\mathcal{A}_p \subset \mathcal{C}_q$ which implies $b \in \mathcal{C}_q \setminus \mathcal{A}_q$.
Let $b' \in \mathcal{A}_q \subset \mathcal{C}_q$ be the block s.t. $Grade_q(b',b)=1$.
This implies $ANS_q(b')>\emptyset$.
Consider two cases: either normal delivery or early delivery happened on node $p$.
Case 1: (normal delivery) $|ANS_p|=n$ means $b' \in \mathcal{C}_p$.
By Lemma  \ref{thm.grade}, $Grade_p(b',b) \neq 1$, contradiction because of $b \in \mathcal{A}_p$ in our assumption.
Case 2: (early delivery) We have $|ANS_p(b)|\geq n-\emptyset$ because $b \in \mathcal{A}_p$.
By the pigeonhole principle, there exists some node $r \in ANS_p(b) \cap ANS_q(b')$.
When the blocks issued by the node $r$ in both local views of nodes $p$ and $q$ are received,
it proves $b \in \mathcal{C}_q$, but $b \notin \mathcal{A}_q$, contradiction.
Thus, the $i=1$ case holds.
For $i+1$ case,
once the output set is the same, the sub-DAG is eventually consistent in every node's
local view. 
Therefore, it follows the statement of Lemmas \ref{thm.ahv}, \ref{thm.grade}, and \ref{thm.arrival} and the case of $i=1$.
This completes the proof.
\end{proof}

\begin{lemma} Let $p$ be a correct node and $\Delta_{BA}$ be the upperbound of processing time of BA, then
the following statements are correct:
\begin{enumerate}
\item If $b$ is received by $p$ at time $T$, then either $b$ has been output or $|ANS_p(b)| = |\mathcal{N}_p \setminus \mathcal{F}_p|$ holds at time $T+\Delta_{BA}+ \lambda$.
\item If $b \in \mathcal{P}$ and $|ANS_p(b)| \geq |\mathcal{N}_p \setminus \mathcal{F}_p|$ holds at time $T$, then the criterion holds during time $T$ and $T+\Delta_{BA}+ \lambda$.
\end{enumerate}
\label{thm.livenesslemma}
\end{lemma}
\begin{proof}
For the first statement, suppose $b$ is not output and $|ANS_p(b)| \neq |\mathcal{N}_p \setminus \mathcal{F}_p|$.
We assume that each correct node receives a block within $\Delta_{BA}+ \lambda$ after it is issued and all correct nodes will ack it or ack the following blocks issued by that node.
Thus, $|ANS_p(b)| \geq |\mathcal{N}_p \setminus \mathcal{F}_p|$.
For the second statement, we have $|ANS_p| \geq |ANS_p(b)| \geq |\mathcal{N}_p \setminus \mathcal{F}_p|$ and any Byzantine node must issue a block during time $T$ and $T+\Delta_{BA}+ \lambda$.
Thus, $|ANS_p|=|\mathcal{N}_p|$ holds, and the criterion must hold by Theorem \ref{thm.sim1}.
\end{proof}

\begin{lemma} If $\mathcal{C}_p \neq \emptyset$ and the criterion holds, then $\mathcal{A}_p \neq \emptyset$ for any node $p$.
\label{thm.outputelm}
\end{lemma}
\begin{proof}
Because $\mathcal{C}_p \neq \emptyset$ and the criterion holds, internal stability implies there exists at least one element in $\mathcal{A}_p$ that precedes the element in $\mathcal{C}_p$.
Thus, $\mathcal{A}_p \neq \emptyset$ holds.
\end{proof}

\begin{theorem} [Liveness for total ordering] 
Let $p$ be a correct node and $\Delta_{BA}+ \lambda$ be the upperbound of processing time of BA. For any time interval $\Delta_o \geq 2\Delta_{BA}+ \lambda$ the criterion holds and the set of output is non-empty in the interval $\Delta_o$.
\label{thm.liveness}
\end{theorem}
\begin{proof}
We first prove that if a block $b$ is received by $p$ at time $T$, the criterion holds and the set of output is non-empty  before $T+ 2\Delta_{BA}+ \lambda$:
by Lemma \ref{thm.livenesslemma}, the criterion will hold before $T+ 2\Delta_{BA}+ \lambda$.
By Lemma \ref{thm.outputelm}, the output set is not empty.
Because each chain has liveness, the theorem is proved.

\end{proof}

\begin{theorem} [Validity of liveness for total ordering] 
For each input of a valid block $b$, the total ordering algorithm will output $b$ eventually. 
\label{thm.validity}
\end{theorem}
\begin{proof}
Assume there exists a valid block $b$ in $p$'s local view s.t. it  is not output by the DEXON total ordering algorithm.
Let the set of blocks that acks $b$ be the set $\Gamma$ and let the set $\Lambda$ be the set of
nodes that issued blocks in $\Gamma$.
Because $b$ is not output, no block in $\Gamma$ will be a candidate block at any time.
Thus, only the elements in set $\Gamma'$ or $\Gamma_s$ could be output, where
$\Gamma'$ is the set of blocks issued by the node in $\mathcal{N} \setminus \Lambda$ and $\Gamma_s$ is the set of blocks whose heights are lower than those of the blocks in $\Gamma$ for each node in $\Lambda$. That is, $\Gamma_s = \Set{b_{q,j} : j<j', \text{where } b_{q,j'}\in \Gamma \text{ and } q \in \Lambda}$.
If the algorithm produces output many times,
we can discover a set of blocks $B \in \Gamma'$ that are acked by some blocks whose heights are higher than those of the blocks in $\Gamma$ for some node in $\Lambda$.
Thus, we have $\#AHV_p(b')<f_{max}$, for all $b' \in B$.
Now, consider two cases: the first case is $\#AHV(b)>\Phi$, then $b \in \mathcal{A}$.
The second case is $\#AHV_p(b)\leq \Phi$, but $|ANS_p| > 2f_{max}$ and $\#AHV_p(b')<f_{max}< \Phi$.
Thus, either $b$ or some blocks in $\Gamma_s$ will be in $\mathcal{A}$.
Because $\Gamma_s$ is a finite set, the algorithm will eventually output block $b$.
\end{proof}

\subsubsection{Complexity}
\label{totocomplexity}
In this section, we analyze the complexity of our total ordering algorithm.
First, we describe the data structure we used:
we store both $AHV'$ and $ANS$ vectors for each block in the pending set,
and $AHV'_p(b)$ is defined as 
\begin{center}
$\begin{cases}
    \bot, & \text{if $q \notin  ANS_p$},\\
    k,    & \text{where } k \text{ is the minimum number s.t. } b_{q,k} \text{ acks } b,\\
    \infty, & \text{otherwise}.
  \end{cases}$\\
\end{center}
We also store a global vector $GAHV_p$, which is the minimum height of blocks issued by each node in $p$'s local view.
Thus, the $AHV_p(b)[j]$ can be computed from $GAHV_p[j]$ and $AHV'_p(b)[j]$ in $\mathcal{O}(1)$.

Second, we start to analyze the algorithm:
when a block is received, if it is a candidate block, the algorithm updates the $AHV$ vectors for the all blocks in the candidate set. This costs $\mathcal{O}(n)$ time. Because no block acks it, only the criterion $|ANS_p|=n$ must be determined.
Otherwise, some block indirectly or directly acks candidate blocks in the $DAG_p$; thus, both $AHV'$ and $ANS$ of all the blocks acked by the received block must be updated.
This costs $\mathcal{O}(n^2)$ time because the total number of blocks in $DAG_p$ is $\mathcal{O}(n^2)$ and the computation of updating for each block is $\mathcal{O}(1)$ .

Third, to maintain the preceding set, a direct method to implement the algorithm is to compute all candidate blocks every time; this method is exactly the same as the pseudo code, but it costs  $\mathcal{O}(n^3)$ because  $|\mathcal{A}|\times|\mathcal{C}\setminus\mathcal{A}|$ blocks  exist and for each iteration, the method requires  $\mathcal{O}(n)$ to compute $|\{j|AHV_p(b)[j]<AHV_p(b')[j]\} |$.
We use a lazy-computation to reduce the complexity because we can store the current value and update it lazily.
Specifically, we use a matrix of $(|\mathcal{A}|+|\mathcal{C}\setminus\mathcal{A}|)$ rows and  $|\mathcal{C}\setminus\mathcal{A}|$ columns, consisting of elements of the form $(i_1,i_2)$, each of which is
$|\{j|AHV_p(b_{i_1})[j]<AHV_p(b_{i_2})[j]\} |$ for $b_{i_1}\in \mathcal{A}$ if $i_1 < |\mathcal{A}|$, $b_{i_1}\in \mathcal{C}\setminus\mathcal{A}$ if $i_1 \geq |\mathcal{A}|$, and 
$b_{i_2}\in \mathcal{C}\setminus\mathcal{A}$.
Thus, every time a block is received, the matrix updates the current relations of $AHV$ functions between candidate blocks. The preceding set and $criterion_p$ can both be determined by the matrix within $\mathcal{O}(n^2)$ time. 
Several operations can be conducted on the matrix.
One operation is updating the matrix when a block is received,
all the elements in the matrix are updated and each update costs a constant operation time.
If the elements in a row all have value less than $\Phi-n+|ABS_p|$, this means that the block corresponding to the row should be  in  the preceding set. 
Thus, the algorithm adds the block into the  preceding set and updates the matrix.
If there exists an elements in a row whose value is larger than $\Phi$, it means that there exists an element in the preceding set that precedes the block corresponding to the row.
Therefore, the complexity of all aforementioned cases is bound by $\mathcal{O}(n^2)$.

Therefore, both time and space complexity of the total ordering algorithm is bound by $\mathcal{O}(n^2)$. 
Note that, no communication exists between any two nodes in  this total ordering algorithm.

\subsection{Timestamping Algorithm}
\label{subsec:timestamping}
In this section, we present the timestamp algorithm, which ensures that
the timestamp is decided from consensus so that the Byzantine nodes cannot bias it.
To compute the consensus timestamp for block $b$, the timestamping algorithm first constructs a vector, whose elements are the newest time of blocks from each chain before block $b$. Then, the median of the vector is the consensus timestamping of block $b$.
The algorithm is shown in Algorithm \ref{algorithm:timestamp}, where $b.Chain\_ID$ means $b$ is generated from the $chain\_ID$ and $b.block\_timestamp$ is the block timestamp of $b$. 
\begin{algorithm}[htbp]
\caption{Compute Timestamp Algorithm}
\label{algorithm:timestamp}
\Fn{Compute\_Timestamp}{
    \Input{ordered chain $Chain_{order} = <b_0, b_1, ... >$ }
    \Output{CompactionChain $\{b_i\}$ with consensus timestamp for each block}
    Let $n$ be the number of chains and $V$ is $n$-dimensional vector\\
    \ForEach{$ID$ of the chain}{
       $V[ID]$ = time of genesis block.
    }
    $i=0$\\
\While{$i \leq$ height of ordered chain}{
    $V[Chain_{ordered}[i].Chain\_ID] = Chain_{ordered}[i].block\_timestamp$\\
    $Chain_{ordered}[i].consensus\_timestamp = median(V)$\\
    $i++$
}
}
\end{algorithm}

\subsection{Consensus on the Blocklattice}
\label{subsec:blocklattice}

\paragraph{Overview}
The DEXON consensus is the consensus of a blocklattice that consists of acking information and numerous single chains as presented in the design in Section \ref{subsection:merge}. With the total ordering algorithm, the blocklattice is processed to a compaction chain by each node individually. The timestamping algorithm computes the consensus timestamp for each block in the compaction chain. 
The parameter in the system contains the number of chains, $\Phi$, threshold $\kappa$ for total ordering, and the network latency bound $\lambda$. We first explain how to add acking information into block information, and then present the mechanism of the consensus of blocklattice including adapting the total ordering and notarizing the blocks so that user can verify the validity of the blocks in the  compaction chain. In the end of this section, we introduce the load balancer that allows each chain to gather the transactions that are not contained in the blocks of other chains. 

\paragraph{Block Information}
To form a blocklattice structure, each block must add additional \textit{acking} information.
When a user proposes a new block, the proposed block \textit{acks} the latest notarized blocks of other chains it has received. This means it sees the blocks from other nodes and supports them.
We say that a block acks another block if its ack field contains the acked block's information.
An \textit{ack field} has the following data structure:

\begin{table}[h!]
\centering
\begin{tabular}{ll}
\toprule
\texttt{block\_proposer\_id} & block proposer's ID \\\addlinespace 
\texttt{acked\_block\_hash} & hash of the acked block \\\addlinespace
\texttt{block\_height} & height of the block, starting at 0 \\
\bottomrule
\end{tabular}
\end{table}

The DEXON blocklattice is formed by numerous single chains with acks between blocks. 
An example of a DEXON blocklattice data structure is illustrated in Figure \ref{figure:blocklattice}.

\begin{figure}
\centering
  \includegraphics[width=0.3\linewidth]{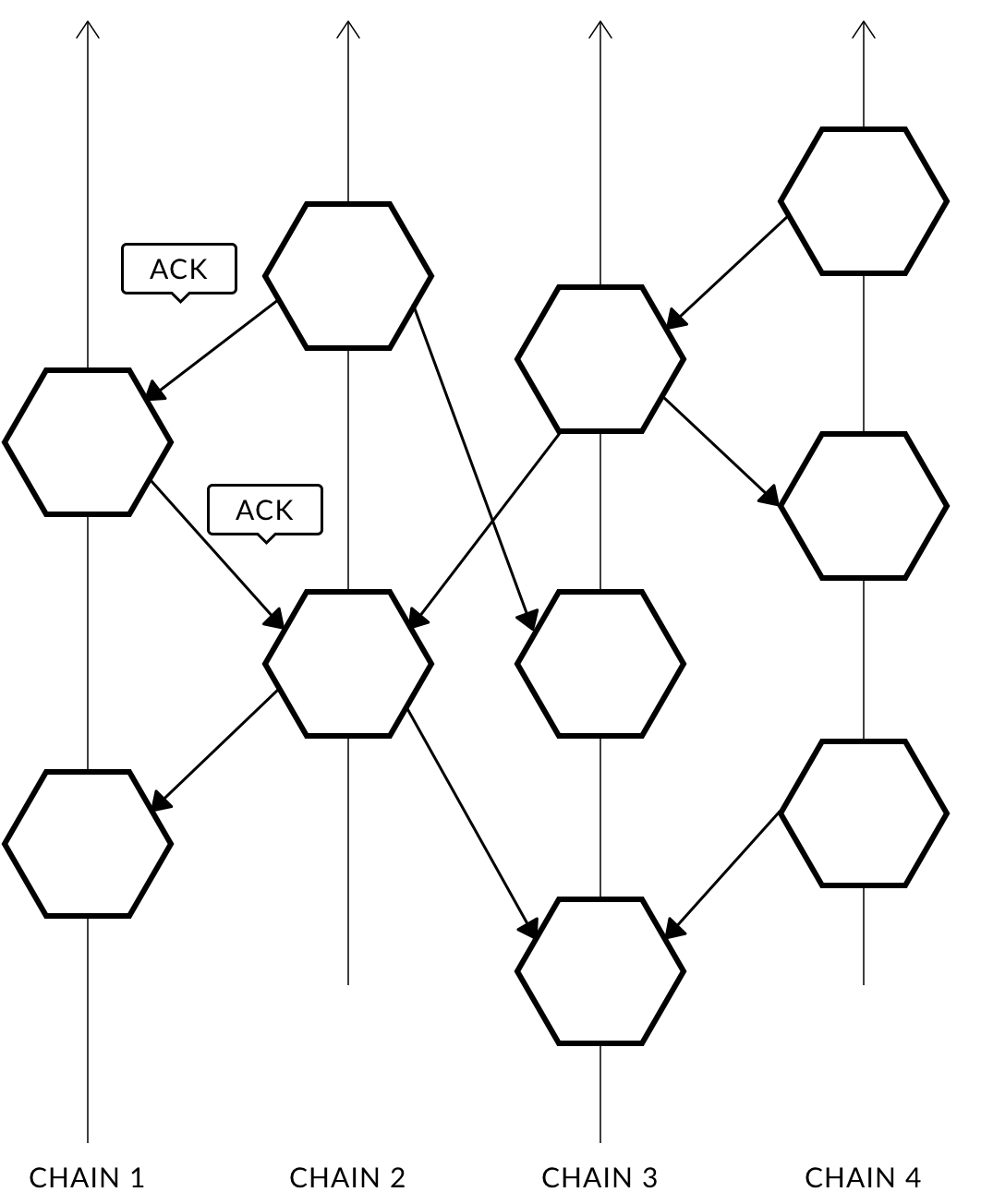}
  \caption{Blocklattice}
  \label{figure:blocklattice}
\end{figure}

\paragraph{Single Chain}
The protocol is described as Section \ref{section:singleChain}. Note that these single chains can share the same CRS, so the system is only required to have one CRS set $\Scrs$.

\paragraph{Compacting into Total-Ordered Chain}
Each single chain generates its blocks individually and these generated blocks form a blocklattice. Then, each node executes the total ordering algorithm with the blocklattice as input, and the output of the total ordering algorithm is the compaction chain, which is compacted from the blocklattice. Then, the system applies the timestamping algorithm to compute the consensus time for each block in the compaction chain.

\paragraph{Notarizing the Blocks in the Compaction Chain}
Once a block is in the compaction chain, the height of the block in the compaction chain, the consensus timestamp, and the threshold signature are confirmed. These three items of information are recorded in one of the next blocks of the compaction chain.
Thus, the validity of the block in the compaction chain can be verified by the following blocks which contains the notarizing information.

\paragraph{Load Balancer}
In DEXON, many single chains grow concurrently. This leads to a problem: some transactions (especially those with high transaction fees) may be packed into several single chains simultaneously. Note that this problem does not do harm to the correctness of DEXON. Suppose a transaction {\sf tx} is notarized in two different single chains. Because all nodes can agree on the same compaction chain after the total ordering algorithm, the transaction in the block with a smaller block height in the compaction chain will be executed and the other will not be executed due to contradiction. However, this problem wastes the space utility of blocks. 

To solve this problem, we introduce a constraint of packing transactions, which is called \emph{load balancer}. Suppose the number of single chains is $N$. Let $S_{\text{notary},j}$ be the notary set of the $j$-th single chain. The node can pack a transaction {\sf tx} into its proposed block only if \[
Hash(\mbox{{\sf tx}}) \mbox{ mod } N = j.
\]
With load balancer, each transaction is only allowed to be packed into one single chain and the problem is resolved.

\subsection{Resistance Against Byzantine Nodes}
We briefly discuss how robust DEXON is against different ratio of Byzantine nodes. As a general framework for combining many single chains, DEXON can resist against $\lfloor (|\Snode|-1)/3 \rfloor$ Byzantine nodes as long as the correctness and the liveness of all individual single chain holds. That is, given $N$ single chains with agreement, DEXON can integrate these single chains into a globally-ordered chain with unbiased timestamp.

As for the single chain algorithm, in Section \ref{subsection:ba}, we show that the agreement and the termination of our Byzantine agreement protocol holds if $t \leq \lfloor (|S_\text{notary}|-1)/3 \rfloor$, where $t$ is the number of Byzantine nodes in $S_\text{notary}$. 

Thus, if we adopt our Byzantine agreement protocol as DEXON's underlining single chain algorithm, we need to guarantee that the number of Byzantine nodes in each notary set is smaller than $\lfloor (|S_\text{notary}|-1)/3 \rfloor$. Let $S_{\text{notary},i}$ be the notary set of $i$-th single chain and $t_i$ is the number of Byzantine nodes in $S_{\text{notary},i}$. Because the notary set of each single chain is chosen independently, the probability that $t_i \leq \lfloor (|S_{\text{notary},i}|-1)/3 \rfloor$ for all $i$ will follow hypergeometric distribution. Consequently, there is a trade-off between the size of notary set, the ratio of Byzantine nodes and the probability of failure. The following table illustrates that the lowest size of $S_\text{notary}$ with different ratio of Byzantine nodes and different probability of failure  given $|\Snode| = $ 10k and $|\Snode| = $ 100k.
 
\begin{table}[h!]
\centering
\begin{tabular}{ccccc}
\multicolumn{1}{c|}{}  & \multicolumn{4}{c}{ratio of Byzantine nodes in $\Snode$}                                      \\ 
\multicolumn{1}{c|}{failure }              & $1/4$ & $1/5$ & $1/4$                 & $1/5$         \\ 
\multicolumn{1}{c|}{probability} & \multicolumn{2}{c}{$\Snode = 10$k}                    & \multicolumn{2}{c}{$\Snode = 100$k} \\ \hline
\multicolumn{1}{c|}{$2^{-40}$}   & ~~~~~1237~~~~~             & 481                        & ~~~~~1402~~~~~        & 489           \\ \hline
\multicolumn{1}{c|}{$2^{-60}$}   & 1789                       & 724                        & 2165                  & 774           \\ \hline
\multicolumn{1}{c|}{$2^{-80}$}   & 2272                       & 952                        & 2900                  & 1054         
\end{tabular}
\caption{The lowest sizes of $S_\text{notary}$ such that less than 1/3 members in $S_\text{notary}$ are Byzantine nodes for $|\Snode|=$ 10k and 100k under the ratio of Byzantine nodes $1/4$ and $1/5$ with different probability are illustrated. For example, if $|\Snode|=$ 10k and the number of Byzantine nodes is $10000\times 1/4 = 2500$, the lowest size of $S_\text{notary}$ is 1237 such that less than 1/3 members in $S_\text{notary}$ are Byzantine nodes with probability of at least $1-2^{-40}$.}
\end{table}

\section{Extension of DEXON}
\label{section:extension}
\subsection{Total Ordering as Sharding Framework}
As mentioned before, compacting a blocklattice into a compaction chain by the total ordering algorithm is a generic method for any blockchain with single-chain structure. A real-time total ordering can support up to 90 nodes in our experiment on a laptop (this experiment does not involve parallelization). A further way to create more sharing chains is to shard many total orderings, and then, to sort the many compacted outputs by each total ordering algorithm according to the unbiased timestamp.

\subsection{Configuration Change}
\label{subsection:configuration}

One of the main features of our consensus is the support of configuration changes while the consensus is operating.
For example, the parameters of the algorithm or the number of chains may be adjusted in order to reach higher efficiency or higher throughput. 
To change configuration, all nodes must initiate at the same state; otherwise, the system will be inconsistent. 
Fortunately, the output of total ordering is a sequence of sets and
each node outputs the same set through either normal delivery or early delivery at the same height of the compaction chain.
Moreover, each block in the compaction chain is marked by a consensus timestamp, which is close to real-world time.
Thus, a way to change configuration is to change the system at the same height or at a specific consensus time (e.g. the first block after the specific consensus time).

\paragraph{Adding/Deleting Chain}
As previously mentioned, the consensus algorithm can change configuration by the state of total ordering. In this section, we demonstrate how to achieve configuration changes by the consensus in each single chain.

Given a specific time $T$, the block $b_i$ after $T$ is the first block after a configuration change in chain $i$.
Let $prev(b)$ be the parent block of $b$; let the blocks before the configuration change be in round $r_1$; and let the blocks after the configuration change be in round $r_2$.
That is, all blocks before $prev(b_i)$ are in $r_1$ and all blocks after $b_i$ are in $r_2$.
Now, we add a rule for acking: any block in round $r_2$ can only ack the last block in round $r_1$ or blocks in round $r_2$.

We use total ordering algorithm to compact the blocks in $r_1$ until the system has received all $\{b_i\}_{i \in \{1,...,n\}}$. 
If some blocks in $r_1$ are not output by total ordering, we use a deterministic topological sort to output them (e.g. output the candidate set with lexicographical order each time). 
Then, we start a new total ordering algorithm with input blocks of round $r_2$.
This is correct because of the liveness of each single chain.
However, the confirmation time may delay due to the quality of synchronization.
The worst case of confirmation time increases as $\lambda + T_{BA} + T_{total\_ordering}$, but configuration changes rarely occur.

For the consensus timestamp, the causality may be demolished when adding or deleting chains because the median of the latest time of each chain may be early. 
However, the number of the blocks whose block time is earlier is very small. Thus, an easy way to fix this is to add a condition $b_i.timestamp = max(b_{i-1}.timestamp,b_i.timestamp)$ that ensures the consensus timestamp to continuously increase.

\bibliographystyle{alpha}
\bibliography{dexon}

\end{document}